\documentclass[aps,prl,twocolumn,showpacs,showkeys]{revtex4}

\usepackage{amsfonts}
\usepackage{amsmath}
\usepackage{graphicx}
\newcommand{\ket}[1]{\, |#1\rangle}
\newcommand{\bra}[1]{\langle #1 |\,}
\usepackage{amsthm}
\newtheorem{thm}{Theorem}

\newtheorem{cor}{Corollary}

\begin{document}

\title{Proof of the Orthogonal Measurement Conjecture for Qubit States}
\date{September 1, 2008}
\author{Andreas Keil}
\affiliation {Department of Physics, National University of Singapore, 2 Science Drive 3, Singapore 117542\\
Centre for Quantum Technologies, National University of Singapore, 3 Science Drive 2, Singapore 117543}
\begin{abstract}
The accessible information of general signal states is obtained by
performing a generalized measurement. In the case that the signal alphabet consists of two states of a qubit system, 
it is proved that a von Neumann (orthogonal) measurement
is sufficient to reach the accessible information.
\end{abstract}
\maketitle

Mutual information is the fundamental quantity of information that two parties, Alice 
and Bob, share. Assume that Alice encodes a message by sending a specific quantum state
 $\rho_r$ 
for each letter in the alphabet of the message. 
The rth letter in the alphabet occurs with probability ${\rm tr} (\rho_r)$ 
in the message. 
Bob sets up a measurement apparatus to determine which state was sent, described by an $n$-outcome positive operator valued measure (POVM) in a $d$-dimensional
space, the range of $\sum\limits_r \rho_r$,
\begin{equation}
\sum\limits_{j=1}^{n} \Pi_j = {\bf{I}}_{d \times d}, \; \Pi_j \ge 0 \;\; \forall j\le n.
\end{equation}
This gives a joint probability matrix 
\begin{equation}
p_{rj} = {\rm tr} \left(\rho_r \Pi_j \right)
\end{equation}
which denotes the probability of Bob to receive outcome $j$ and Alice to send $l$. 
We denote the marginals of the distribution by 
\begin{equation}
p_{r \cdot} =\sum\limits_j p_{rj},\; \;
p_{\cdot j} =\sum\limits_r p_{rj}.
\end{equation}

The mutual information 
\begin{equation}
I(\rho,\Pi) = \sum\limits_{r,j} p_{rj} \log \left(\frac{p_{rj}}{p_{r \cdot} p_{\cdot j}} \right)
\end{equation}
tells us how much information on average is transmitted to Bob per sent state \cite{Nielsen}. 
Choosing a POVM that maximizes the mutual information defines an important quantity;
this is
called the accessible information
\begin{equation}
 I_{{\rm acc}}=\max\limits_{\Pi} I(\rho,\Pi).
\end{equation}

The question which POVM maximizes the mutual information, which
was raised by Holevo in 1973 
\cite{Holevo},  is generally unanswered and 
usually addressed numerically \cite{Berge,Jun,Somim}. Even the simpler question
of how many outcomes are sufficient is unanswered. 
It has been shown \cite{Holevo2} that 
an orthogonal (von Neumann) measurement, i.e.
\begin{equation}
\Pi_l \Pi_k = \delta_{lk} \Pi_k, \;\; {\rm tr} \{\Pi_k\} =1,
\end{equation}
is in general not sufficient. 
Levitin \cite{Levitin} 
conjectured in 1995 that if Alice's alphabet consists of up to $d$ 
states, an orthogonal measurement 
is sufficient. If so, the number of outcomes would be equal to 
the dimension of the Hilbert space. This conjecture has been shown not to
hold in general by Shor \cite{Shor}. 
In the same paper Shor reported that Fuchs and Perez
affirmed numerically that if the alphabet consists of two states the optimal
measurement is an orthogonal measurement. This is the 
\emph{orthogonal measurement conjecture}. 
For two pure states it was proved to be true in arbitrary dimensions 
by Levitin \cite{Levitin}.

This conjecture has important experimental and theoretical implications.
In an experiment, orthogonal measurements
are generally easier to implement than arbitrary generalized measurements. 
From a theoretical 
point, knowing the accessible information is crucial 
to determine the $C^{1,1}$-channel capacity \cite{Shannon} and
for security analysis
using the Csisz{\'a}r-K\"orner theorem \cite{Csiszar},
for example the thresholds for a collective attack on the Singapore protocol
\cite{Singapore} are obtained by determining the accessible information.
 Work has been
done under the assumption that this conjecture is true \cite{Fuchs}.
In the sequel we will prove this conjecture for qubits. 

Label the two states of the qubit system by $\rho_1$ and $\rho_2$. It is always
possible to find a basis in which their matrix representation is real.
 Hence, we can conclude from Sasaki \emph{et al.} \cite{Sasaki} in an extension of the proof by Davies \cite{Davies}  
that the maximum number of rank-1
 outcomes needed to reach the accessible information is three. In the 
same paper Sasaki showed that these outcomes can be chosen to be real as well. The outcomes can
be written $\Pi_j = \ket{j} \bra{j}$ where in general $\bra{j} j\rangle < 1$.
What we will show in the following is that two of these vectors have
to be linearly dependent if the mutual information is maximal. 
The most general variation of these outcomes is given by
\begin{equation}
 \delta \ket{j} = i \sum\limits_m \ket{m} \epsilon_{mj}, \;\; \epsilon^*_{mj} = \epsilon_{jm}
\end{equation}
with $\epsilon$ representing an arbitrary hermitian three by three matrix. With this variation, the first order
variation of the mutual information is given by
\begin{eqnarray*}
\delta I &=&\sum\limits_{r,j} \delta p_{rj} \log \frac{p_{rj}}{p_{\cdot j}}\\
&=& -2 \sum\limits_{r,j,m} \mathrm{Im} \left( \bra{j} \rho_r \ket{m} \epsilon_{mj}\right) \log \frac{p_{rj}}{p_{\cdot j}}. 
\end{eqnarray*}
For each pair of outcomes $(k,l)$ we can set $\epsilon_{mj}$ for 
$\{k,l\}\neq\{m,j\}$  to zero, except for $\epsilon_{kl}=\frac{i}{2}$ 
and $\epsilon_{lk}=-\frac{i}{2}$.
This gives us the following set of variations 
\begin{equation}
 \delta_{(k,l)} I = \sum\limits_{r=1}^{2} 
\bra{k} \rho_r \ket{l} \log\left(\frac{p_{rk}}{p_{rl}}\frac{p_{\cdot l}}{p_{\cdot k}} \right)=0.\label{keyeq}
\end{equation}
Since these sets are antisymmetric in $k,l$ we get exactly three independent pairs. Fix 
one of the directions, say $\ket{1}$, and one vector $\ket{0}$ orthonormal
to $\ket{1}$ 
 to complete a basis of our
real Hilbert space. Any vector $\ket{n}$ can be expressed as 
\begin{equation}
 \ket{n}= \beta_0 (n) \ket{0} + \beta_1(n) \ket{1}. \label{basis}
\end{equation}
We want to see what are the restrictions from these equations 
(\ref{keyeq}) on the vectors. 
When $\beta_0$ is zero the conjecture is trivially true, since the vector
would be proportional to $\ket{1}$. 
 Observe that, due to scale invariance of (\ref{keyeq}), it is always possible 
to divide out $\beta_0\neq 0$ and restrict ourselves 
to $\ket{n}=  \ket{0} + t \ket{1}$ with $t$ an arbitrary real number.

In the following we assume neither $\rho_1$ nor $\rho_2$ to be pure. We will
outline how to approach this case at the end of the paper. We get: 
\begin{eqnarray}
0\!=\!\bra{1}\!\left(\!\rho_1\!+\!\rho_2\!\right)\!\ket{1}\!\label{teqn}
\!\sum\limits_{r=1}^{2}\!\alpha_r Q_r^\prime(t)\!\log 
\frac {Q_r(t)}
{\alpha_1\!Q_1(t)\!+\!\alpha_2\! Q_2(t)} 
\end{eqnarray}
with
\begin{equation}
 Q_{r}=t^2 + 2 t \frac{\bra{1}\rho_{r}\ket{0}}
{\bra{1}\rho_{r}\ket{1}} + \frac{\bra{0}\rho_{r}\ket{0}}{\bra{1}
\rho_{r}\ket{1}}, \hspace{0,5em} 
\alpha_r = \frac{\bra{1}\rho_r\ket{1}}{\bra{1}\left(\rho_1 + \rho_2\right) \ket{1}}
\end{equation}
and \emph{prime} denoting differentiation with respect to $t$.
Introducing
\begin{equation}
\xi_{r}={\frac{\bra{1}\rho_{r}\ket{0}}{\bra{1}\rho_{r}\ket{1}}}, \hspace{0,5em}
\eta_{r}=\frac{\bra{0}\rho_{r}\ket{0}}{\bra{1}\rho_{r}\ket{1}},
\end{equation}
the range for these variables is restricted due to positivity of the states to
\begin{equation}
 0 \le \xi_r^2 < \eta_r < \infty, \;\; 0 < \alpha_r < 1, \;\;  r=1,2, \;\;
\alpha_1+\alpha_2=1.  \label{constraints}
\end{equation}

The following theorem is the cornerstone for handling the general case:
\begin{thm}\label{thmzeros}
Each function defined by
\begin{equation}
f_{(\alpha,\xi,\eta)}(t) = \sum\limits_{r=1}^{2} \alpha_r Q_r^\prime(t) \log 
\frac {Q_r(t)}
{\alpha_1 Q_1(t) + \alpha_2 Q_2(t)} \label{function2}
\end{equation}
with constraints given by (\ref{constraints}), 
$Q_r(t)= t^2 + 2t \xi_r + \eta_r$  
and $Q_r^\prime(t) = 2(t+ \xi_r)$, 
has at most two real roots.
\end{thm}
We will prove this later; first, we use this theorem to show our central result:
\begin{thm}
 Every local maximum of the mutual information is a von Neumann measurement.
\end{thm}
\begin{proof}
 Assume that the mutual information is stationary and that POVM is not von Neumann. Observe that in (\ref{function2}) if one logarithm is zero, automatically the other
is zero as well. Since $\ket{2}$ and $\ket{3}$ have to be distinct, Theorem~\ref{thmzeros}
tells us that one of these states must set the logarithm to zero, say $\ket{3}$. This means that 
\begin{eqnarray*}
 \frac{p_{\cdot 1} p_{13}}{p_{\cdot 3 } p_{11}} = 1 \; \leftrightarrow \frac{p_{21}}{p_{11}}=
\frac{p_{23}}{p_{13}},
\end{eqnarray*}
so outcome 1 and outcome 3 are equivalent.
Since the same reasoning is applicable to $\bra{2}$ instead of $\bra{1}$ in equations 
(\ref{keyeq}),(\ref{teqn}) we find that all outcomes are equivalent and we are in a minimum. 
\end{proof}
\begin{figure}[htb!]
\centering
\includegraphics[width=9.5cm]{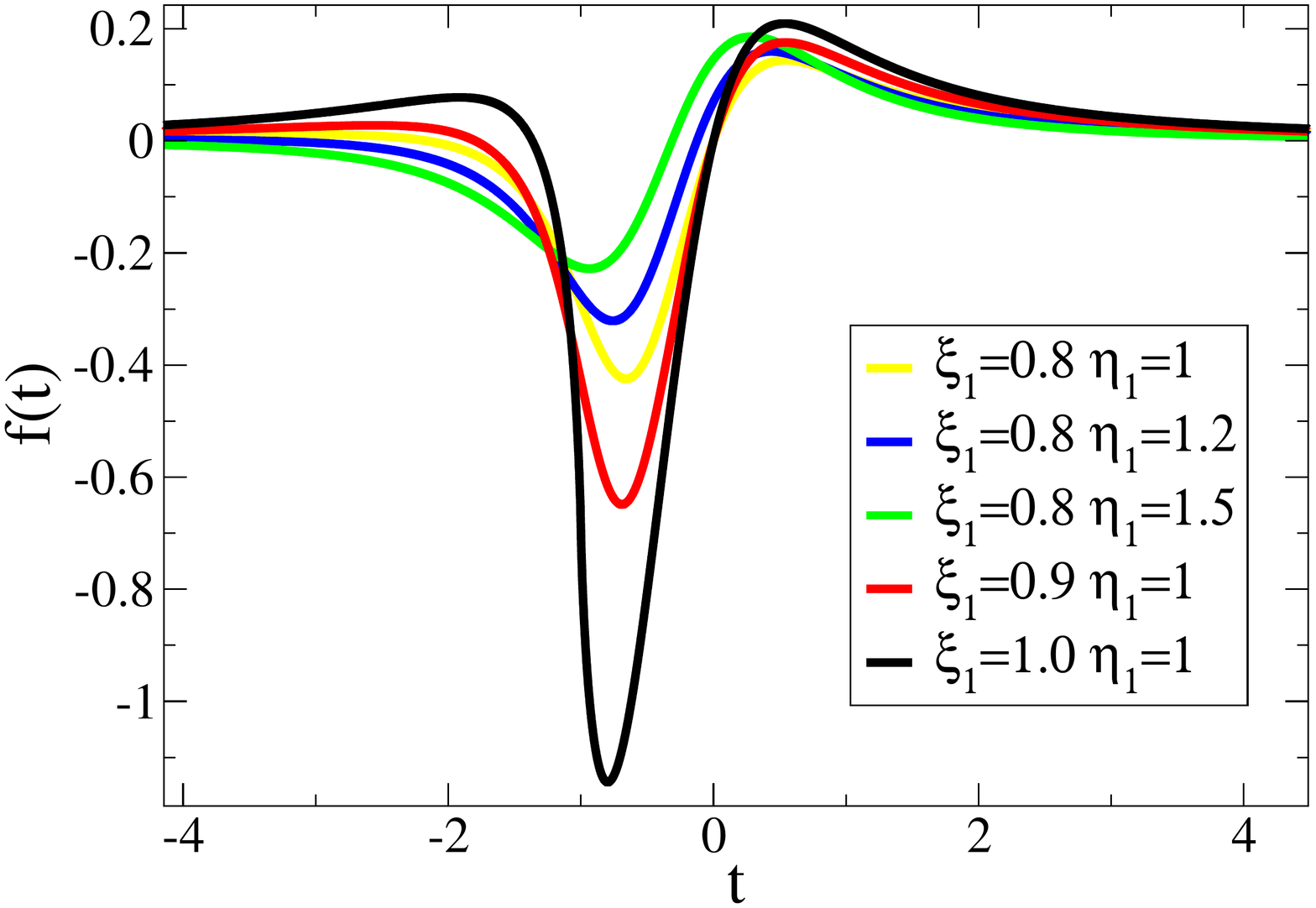}
\caption{Illustrating the function $f(t)$ in Theorem (\ref{thmzeros}) for 
selected values, $\alpha_1~=\alpha_2~=~0.5$, $\xi_2=0$ and $\eta_2=1$, other parameters as stated in the graph.}
\label{figtypical}
\end{figure}

The remainder of this paper will be used to show some lemmata to proof Theorem~\ref{thmzeros}.
The graph of $f$ is plotted in  Fig.\ref{figtypical} for selected values of the parameters  to
illustrate its shape. 
Using Landau little-o notation and expansion of the logarithm in 
(\ref{function2}) shows that asymptotically  
\begin{equation}
 f(t) = 0 + o(1/t).
\end{equation}
In particular, this implies that the function vanishes at infinity. 
Fortunately the term in front of the logarithm is linear, so 
by calculating the second derivative we arrive
at a rational function. Introducing the abbreviations
\begin{equation}
 L=Q_1^\prime Q_2 - Q_2^\prime Q_1, \;\; Q_s=\alpha_1 Q_1 + \alpha_2 Q_2,
\end{equation}
for the first and second derivative we get 
\begin{eqnarray}
f^{\prime} &=& 2 \left(\sum\limits_{l=1}^2 \alpha_l \log \frac{Q_l}{Q_s}\right)+
\alpha_1 \alpha_2 \frac{\left(Q_1^\prime Q_2 - Q_2^\prime Q_1 \right)^2}{Q_s Q_1 Q_2} \nonumber\\
f^{\prime \prime} &=&
\alpha_1 \alpha_2 \frac{L}{(Q_1 Q_2 Q_s)^2} P  \label{secondd}
\end{eqnarray}
with
\begin{equation}
P = 3 L^\prime (Q_1 Q_2 Q_s) - (Q_1 Q_2 Q_s)^\prime L. \label{P}
\end{equation}
The second derivative of $f$ is a product of a rational function with a maximum of two zeros and no poles
-- since $L$ is of second order and $Q_1 Q_2 Q_s$ strictly positive --
 and a sixth-order polynomial $P$. 
 In the following we will show that the number of real roots of $P$ is limited to two, which implies 
that the number of inflection points of $f$ is limited by four. 
Since $f$ approaches zero at infinity, it follows that the
number of real roots of $f$ is at maximum two, completing the proof of \emph{Theorem 1}.

Since the number of roots is translation and scaling invariant, it is always possible to set $\xi_2 =0$ and
$\eta_2=1$. It will be convenient in the following to label $\xi_1$ simply as $\xi$ and to 
introduce the difference between $\eta_1$ and $\xi_1^2$,
\begin{equation}
 X = \eta_1 - \xi ^2 > 0
\end{equation}
which is positive because of (\ref{constraints}).
Since a direct attack on this family of polynomials does not prove successful, we will look at its 
discriminant, showing that it is non-zero which implies that on 
the whole parameter domain no double root occurs.
 Since the coefficients (and therefore the roots) of our set of polynomials depend continuously
on the parameters, any change in the number of roots occurs either by going through a double root or 
changing the degree of the polynomial. In particular we will use the following
\begin{thm}
Let G be a family of real valued polynomials with formal degree $n$, i.e.
\begin{equation*}
 G: \mathbb{R}^m \rightarrow \mathbb{R} [x], z \mapsto G_z = \sum\limits_{k=0}^{n} g(z)_k x^k
\end{equation*}
which is continuous in the coefficients of $x$, and $D$ a connected domain of $G$ for which 
the discriminant of $G$ does not vanish. The number of roots of $G_z$ is constant on each of the following
domains
\begin{eqnarray*}
 D_0 &=& \{ z \in \mathbb{R}^m | g(z)_n = 0 \},\\
 D_1 &=& D \backslash D_0,
\end{eqnarray*}
and the number of roots on $D_1$ is the number of roots on $D_0$ plus one.
\end{thm}
\begin{proof}
 Assume that the discriminant does not vanish. This implies that the polynomials
on $D_1$ do not have any double root. Since the roots of the polynomials depend
continuously on its coefficients this implies that the polynomials of 
each connected component of $D_1$ have the same number of real roots. 
The non-vanishing of the discriminant on $D_0$, where the highest coefficient
vanishes, implies that the second highest coefficient is non-zero and
that no double root occurs as well. This fixes the number of real roots on each
connected component of $D_0$ as well, and there is exactly one less which
went off to infinity. 
\end{proof}

\begin{cor}
 The discriminant $\Delta(P,t)$ of P is non-vanishing for all $\xi^2 > 0$ with $X>0$ 
and $0 \leq \alpha_1 \leq 1$,
and in the case that $\xi^2=0$ for all $0 <  X \neq 1$ 
and $0 \leq \alpha_1 \leq 1$.
\end{cor}
\begin{proof}
The discriminant of $P$ is given by
\begin{eqnarray*}
\Delta(P,t)&=& 589824 X \left[(1-X-\xi^2)^2 + 4 \xi^2 \right] ^7 \\
&&\times\left\{- \left[\alpha_1 (\alpha_2 \xi^2   +1)+ \alpha_2 X\right] \right\}
P_2 (\alpha_1, X,\xi^2)^2
\end{eqnarray*}
All terms except the last are obviously nonzero, so we take a closer look at the last term, which is a fourth order
polynomial in $X$, 
\begin{eqnarray}
P_2(\alpha_1,X,\xi^2) =\sum\limits_{k=0}^{4} Y_k(\alpha_1,\xi^2) X^k
\end{eqnarray}
We are now going to show that each of the coefficients is non-negative and at least one of them is non-vanishing,
giving us a positive polynomial.
The coefficient
\begin{equation}
Y_4 =\alpha_2^2 (16 \,\alpha_2 + \alpha_1^2) 
\end{equation}
is obviously zero for $\alpha_1=1$, since $\alpha_2=1-\alpha_1$,
 and positive otherwise.
The coefficient
\begin{eqnarray*}
 Y_3 &=&  -4 (\alpha_2)^2 (3 \alpha_1^2 + 4\alpha_1 -8) \xi^2 \\
&&+ 4 \, \alpha_2 \, \left(-3\alpha_1^3
 + 67 \alpha_1^2 -196 \alpha_1   +136 \right)
\end{eqnarray*}
is affine in $\xi^2$. To show that this coefficient is greater than zero, we
use that
\begin{gather*}
 -3\alpha_1^2 - 4\alpha_1 +8 \geq - 3 \cdot 1^2 - 4 + 8 = 1
\end{gather*}
and
\begin{gather*}
 -3 \alpha_1^3 + 67 \alpha_1^2 - 196 \alpha_1 + 136 > -3 + 66 \alpha_1^2  - 198 \alpha_1 + 136 \\
= 66 ( \alpha_1 ^2 - 3 \alpha_1 +2 ) +1=  66 (2- \alpha_1)(1-\alpha_1) \ge 1
\end{gather*}
we have 
\begin{gather*}
 Y_3 \ge 4 \, \alpha_2 \,(\alpha_2 \xi^2 + 1) \ge 0
\end{gather*}
with $Y_3=0$ only if $\alpha=1$.
The remaining coefficient can be treated in a similar fashion so we will just state them for completeness.

{\bf Coefficient $Y_2$:} $Y_2$ is a quadratic polynomial in $\xi^2$ :
\begin{eqnarray*}
Y_2 &= & 2 (-13 \alpha_1^4 +34 \alpha_1^3 - 21 \alpha_1^2 - 8 \alpha_1 +8) \xi^4 \\ 
& - &
2 (122 \alpha_1 ^4 - 636 \alpha_1^3 + 914 \alpha_1^2 -384 \alpha_1 -16) \xi^2 \\
& - & 2 \left(13 \alpha_1^4 - 26 \alpha_1^3 + 405 \alpha_1^2 - 392 \alpha_1 -8\right)
\end{eqnarray*}
in which all terms can be shown to be non-negative.
 
{\bf Coefficient $Y_0$:} The last term is given by
\begin{eqnarray*}
&&Y_0=  \alpha_1^2 (1+\xi^2)^2  \\
&&\times \left[ \xi^4 ( \alpha_2)^2 + \xi^2 \, 2(1-\alpha_1^2 + 6\, \alpha_1 \alpha_2)
+1+\alpha_1^2 + 14 \alpha_1\right]\\
 &&\ge 0 
\end{eqnarray*}
This shows that the discriminant is non-zero.
\end{proof}
To determine the number of roots we need to look at one polynomial for 
conveniently chosen parameters. Choose
\begin{equation*}
\xi=2, \;\; X=1, \;\; \alpha=\frac{1}{2}
\end{equation*}
giving us for (\ref{P}) 
\begin{equation*}
P= -64 (1 + t) (7 + 8 t + 8 t^2 + 4 t^3 + t^4)
\end{equation*}
which has exactly one real root. This completes the proof for mixed states. 
In the case that one of the states is pure, i.e. $\xi^2=\eta$, the orthogonal measurement conjecture follows from continuity in all variables of the mutual information.  
To show the stronger statement of Theorem 2 one has to  perform
a similar analysis as the one done in this paper,
with the difference, that the second derivative 
(\ref{secondd}) has a maximum of three roots and one simple pole.

In this paper we proved the orthogonal measurement conjecture for general 
two qubit states. It is natural to ask if this proof can be extended to higher
dimensions. Already for qutrits we are facing serious difficulties. Two
qutrits can in general not be transformed to real matrices with only one
unitary transformation. Restricting ourselves to real qutrit or qunit matrices
we can derive equations of similar type as equations (\ref{keyeq}). Performing
a similar expansions of the vectors in a basis as in (\ref{basis}) does
not give us, as in this paper, functions of one real variable but gives 
rise to intersections of transcendental curves in 
projective space, leading us into a vast, unexplored territory.

I want to thank Berge Englert, Lai Choy Heng, Frederick Willeboordse,
Jun Suzuki and Alexander Shapeev for valuable discussions and for
 making this work possible.
This work is supported by the National Research Foundation \& Ministry of Education, Singapore.

\bibliography{orthogonal}{}
\bibliographystyle{apsrev}

\end{document}